\documentclass[preprint,report]{elsarticle}

\usepackage{array,xspace,multirow,hhline,tikz,colortbl,tabularx,booktabs,fixltx2e,amsmath,amssymb,amsfonts,amsthm}
\usepackage{paralist}
\usepackage[boxed]{algorithm}
\usepackage[noend]{algorithmic}
\usepackage{amssymb,amsthm,xspace}

	\newcommand{\ceil}[1]{\lceil #1 \rceil }
	\newcommand{\floor}[1]{\lfloor #1 \rfloor }
\usepackage{verbatim,ifthen}
\usepackage{enumitem}
\usepackage{url}
\usepackage{pifont}
\usepackage{ifthen}
\usepackage{calrsfs,mathrsfs}
\usepackage{bbding,pifont}
\usepackage{pgflibraryshapes}
\usetikzlibrary{arrows}
\usepackage{centernot}
\usetikzlibrary{decorations.pathreplacing}
\usetikzlibrary{patterns}
\usepackage{pgflibraryshapes}
\usetikzlibrary{positioning,chains,fit,shapes,calc}
	\usepackage{varioref}
	\usepackage{nicefrac}
	\usepackage{hyperref}
\usepackage{cleveref}
\usepackage{tikz}
\usetikzlibrary{topaths,calc}

	\usepackage{tabularx}
	\usepackage{mathtools}

	\algsetup{linenodelimiter=\,}
	\algsetup{linenosize=\tiny}
	\algsetup{indent=2em}

\definecolor{light-gray}{gray}{0.9}

\newtheorem{definition}{Definition}%

\usepackage{boxedminipage}
\usepackage{xspace}

				\newcommand{\bt}[1][]{\ensuremath{\ifthenelse{\equal{#1}{}}{\mathit{BT}}{\mathit{BT}(#1)}}\xspace}

	\newtheorem{proposition}{Proposition}%
	\newtheorem{example}{Example}
	
	\newtheorem{remark}{Remark}
	
\usepackage{eqparbox}

\DeclareMathOperator*{\argmin}{arg\,min}

	\usepackage{enumitem}
	\setenumerate[1]{label=\rm(\it{\roman{*}}\rm),ref=({\it\roman{*}}),leftmargin=*}
	\newlength{\wordlength}

\usepackage{enumitem}
\setenumerate[1]{label=\rm(\it{\roman{*}}\rm),ref=({\it\roman{*}}),leftmargin=*}

\newcommand{\nbh}[1][]{
	\ifthenelse{\equal{#1}{}}{\nu}{\nu(#1)}
}

\newcommand{\cstr}[1][]{
	\ifthenelse{\equal{#1}{}}{\mathscr S}{\cstr(#1)}
}

\newcommand{\choice}[1][]{
	\ifthenelse{\equal{#1}{}}{\mathit{C}}{\choice(#1)}
}

\newcommand{\nimrodout}[1]{}

\begin{document}

\title{Proportionally Representative Participatory Budgeting:\\ Axioms and Algorithms}

	\author{Haris Aziz}
	 \ead{haris.aziz@data61.csiro.au}
      \address{Data61, CSIRO and UNSW Australia\\
Sydney, Australia}

	\author{ Barton Lee} \ead{barton.e.lee@gmail.com}

	\address{Data61, CSIRO and the School of Economics, UNSW,\\ Sydney, Australia}
	
	\author{Nimrod Talmon} \ead{nimrodtalmon77@gmail.com}

	\address{Weizmann Institute of Science, Israel}
	
%

\begin{abstract}
Participatory budgeting is one of the exciting developments in deliberative grassroots democracy.
We concentrate on approval elections and propose proportional representation axioms in participatory budgeting,
by generalizing relevant axioms for approval-based multi-winner elections.
We observe a rich landscape with respect to the computational complexity of identifying proportional budgets and computing such,
and present budgeting methods that satisfy these axioms by identifying budgets that are representative to the demands of vast segments of the voters. 
\end{abstract}

\maketitle

\section{Introduction}\label{section:introduction}

\begin{quote}
``Participatory budgeting (PB) has become a central topic of discussion and significant field of innovation for those involved in democracy and local development''--- ~Cabannes~\cite{cabannes2004participatory}.
\end{quote}

\begin{quote}
``Whatever the best approach to participatory budgeting is, now is the time to identify it, before various heuristics become hopelessly ingrained''---~\citet{benade2017preference}.
\end{quote}

Participatory budgeting (PB),
used in hundreds of cities across several continents (especially South America)\footnote{PB is used in several cities in the USA~\cite{Gilm12a}. Paris is organizing one of the largest citywide participatory budgets (\url{https://budgetparticipatif.paris.fr/bp}).},
 is a grass-root deliberative approach where common people make public budgeting decisions.
One of its primary advantages is ``the more transparent management and more accessible municipal process that it allows''~\cite{cabannes2004participatory}.
PB is very successful, however most methods currently used do not take into account a formal approach to proportionality. The lack of sufficient representation of key groups
in participatory budgeting can be a critical shortcoming as has been witnessed in a participatory budgeting program in Porto Alegre, Brazil~\cite{BRTK03a}. We undertake a formal approach to representative participatory budgeting in which we propose both representation axioms and rules for achieving representative budgets.
As it has been observed that the exact way PB is implemented is critical to its success 
such a principled axiomatic and rule-based approach may especially be useful. The problem of finding the right approach to solve PB has been mentioned as a ``\textit{grand challenge for computational social choice, especially at a point in the field's evolution where it is gaining real-world relevance by helping people make decisions in practice}''~\citep{benade2017preference}.

Formally, we view participatory budgeting as a generalization of \textit{multi-winner elections},
specifically since in multi-winner elections usually we aim at selecting $k$ candidates out of a given set of $m$ candidates,
while in participatory budgeting each of the given candidates (usually referred to as projects or items) does not come at unit cost,
but has its own cost. The task is then
to find a satisfactory set of projects (notice, not specifically $k$ of them) whose total cost does not go beyond some specified budget limit $L$. The task is achieved while taking into account the preferences expressed by the electorate.

Most methods (see Section~\ref{section:related work}) of participatory budgeting do not take into account aspects of proportionality;
for example, one popular method uses $k$-Approval, by letting each voter specify a set of $k$ projects of her liking,
then order the projects by decreasing number of (sum of) approvals,
and greedily go over the list, taking (as a ``committee member'') each project which does not cause the proposed expenditure to go over the given limit. 

In certain situations,
however,
it is desirable to find budgets which are more proportional.
For example, if deciding on building, say, schools or recreational parks in a city,
the number of schools per neighborhood shall be roughly proportional to the population;
specifically, building schools only in the city center--where many residents live--is usually not an accepted outcome.

Luckily,
in the study of multi-winner election (i.e., in literature on committee selection) there are quite a number of ideas on how to achieve such proportionality;
specifically, several axioms of proportionality have been devised,
together with several voting rules satisfying these axioms.
Our approach for formalizing and achieving proportionality in PB is thus
to explore ways of generalizing the axioms formulating proportionality in multi-winner elections;
in effect,
``lifting'' the knowledge from committee voting / multi-winner voting to participatory budgeting (where the latter generalizes the former, by considering various costs for the candidates).

\paragraph{Contributions}
We initiate a formal axiomatic approach to participatory budgeting in which each project is either selected or not selected,
each voter approves a subset of the available projects, and the goal is to identify a proportionally representative budget. 

We present a series of proportional representation axioms which generalize the \textit{JR (justified representation)} and \textit{PJR (proportional justified representation)} axioms considered in approval-based multi-winner voting~\cite{aziz2017justified,sanchez2017proportional} . Our quest for the `right' axioms 
is motivated by the following main factors: (1) they should have normative justification in representing large-enough and cohesive groups of voters; (2) they should be strong enough to identify particularly representative outcomes but not too strong so that for some instances no satisfying outcomes exist; and (3) they should be computationally amenable, ideally admitting a polynomial-time algorithm. 

Our study includes a detailed examination of promising proportional representation axioms, computational results on finding and testing representative budgets, as well as identifying the logical relations between the axioms. In practical terms, we propose several algorithms that find suitable budgets with axiomatic guarantees of proportionality. One of the most compelling algorithms we present is a careful generalization of Phragmen's unordered rule (also referred to as Phragmen's sequential rule~\cite{BFJL16a,Jans16a}).
Most of our results are summarized in Table~\ref{table:summary}.

\subsection{Related Work}\label{section:related work}

Participatory budgeting~\cite{cabannes2004participatory} is concerned with letting citizens decide upon the way their collective funds are being used.
We discuss some methods by which this is done.
Knapsack voting, which is similar to $k$-Approval, but where each voter approves a set of projects whose total cost does not go over the given limit,
is considered by Goel et al.~\cite{goel2015knapsack}.
In an earlier paper, Klamler et al. \cite{KPR12a} presented algorithms for committee selection with knapsack constraints.
Both these works have utilitarian concerns and do not capture proportional representation. 
Further preference elicitation methods are studied by Benade et al.~\cite{benade2017preference},
that allow voters to either rank the candidates by their value (\emph{Value voting}) or value-for-money (\emph{Value-for-money voting}),
and so-called \emph{Threshold voting}, where each voter specifies the subset of projects whose value is perceived to be above a predefined threshold.
Shapiro and Talmon~\cite{condorcetbudgeting} generalize Condorcet's principle to PB
and devise a polynomial-time budgeting method to compute such budgets; these budgets are majoritarian in nature,
while here we are concerned with budgets satisfying proportional representation.

Fain et al.~\cite{fain2016core} study proportional representation in participatory budgeting,
thus their work is closely related to ours.
In their model, however, the task is to decide the amount of funds to spend on each project,
while our model is more discrete, as we aim at deciding which projects to fund.
Indeed, some projects are inherently indivisible or discrete in nature. For example, one can decide to fund one road or two roads, say, but not a road an a half.
\citet{ABM17a} discuss how a probabilistic approach to voting can be used to address PB. Just like the work of \citet{fain2016core}, the approach applies to projects that are `divisible' in nature. 
Another related paper is by Conitzer et al.~\cite{conitzer2017fair},
who consider proportionality issues in fair division problems.
There, however, proportionality is considered at the level of individuals and cohesive groups of voters are not considered. 

The study of multi-winner elections~\cite{chaptermw} is a thriving sub-field of computational social choice~\cite{moulin2016handbook},
and is concerned mainly with aggregation method for committee selection in which each candidate has a unit cost.
Some papers exist which study election scenarios with different costs for different candidates,
such as the paper of Lu and Butilier~\cite{lu2011budgeted}, who consider a generalization of the Chamberlin--Courant rule~\cite{chamberlincourant}.
%
There are quite a number of papers within the sub-field of multi-winner elections which concentrate on issues of proportionality,
some of which are mentioned next.
Aziz et al.~\cite{aziz2017justified} consider proportionality axioms for approval-based multi-winner elections. 
Here, we generalize this line of work to PB.
There are several follow up papers on approval-based multi-winner~(see e.g., \citep{BFJL16a,AzHu16a,sanchez2017proportional}).
The paper of Aziz et al.~\cite{aziz2017condorcet} considers multi-winner rules extending (in certain ways) Condorcet's principle,
and demonstrates that some of them satisfy certain axiomatic properties aiming at proportionality.

\section{Preliminaries}


Let $C$ be a set of items and let $w: C \rightarrow \mathbb{R}$, $c \mapsto w(c)$, be the associated \emph{cost function}.
We normalize the cost function such that $\min_{c\in C} w(c)=1$. This assumption is without loss of generality and assists in keeping our axioms invariant to scalings of the currency. Abusing notation slightly, given a subset of items $C'\subseteq C$ we define $w(C')=\sum_{c\in C'} w(c)$.
A \emph{budget limit} will be denoted by $L$; a \emph{budget} $W$ is said to be \emph{feasible} if $w(W)\le L$.

Let $V$ be a set of voters, where each voter $i\in V$ submits an approval ballot $A_i\subseteq C$,
which is an unranked ballot of items in $C$ which they approve of.
The vector of approval ballots, called a \emph{ballot profile}, is denoted by $A=(A_1, \ldots, A_n)$.
A set of voters $V'\subseteq V$ is said to be cohesive if 
$\cap_{i\in V'} A_i\neq \emptyset$,
that is, if they unanimously agree/support some item(s) $c\in \cap_{i\in V'} A_i$.

The goal of a \emph{budgeting method} is to take as an input a ballot profile $A$ and produce an output/budget $W\subseteq C$ which is feasible and satisfies some desirable axioms.
First, we would like our budgets to not `leave money on the table';
we formalize it as follows.

\begin{definition}[Exhaustiveness]
Given a budget limit $L$ and budget $W$, the budget is said to be \emph{exhaustive} if for all $c\notin W$ it holds that
$$w(W\cup\{c\})=w(W)+w(c)>L.$$
If this is not the case, the we say that the budget $W$ is \emph{non-exhaustive}.
\end{definition}

The second class of axioms,
discussed in Section~\ref{section:axioms},
relate to justified representation in the budget outcome $W$ with respect to the ballot profile $A$.
%
%
%
Next, we point out how the setting we consider generalizes approval-based committee voting.

\paragraph{Committee Voting / Multi-winner Voting}

In approval-based multi-winner election we have a set of candidates $C$ and a set of voters $V$,
where each voter $V$ corresponds to a subset of $C$, consisting of her approved candidates.
The task is to select a committee $S \subseteq C$ of $k$ candidates. Observe that, if $L=k$ and $w(c)=1$ for each $c\in C$, then our setting coincides with approval-based committee voting / multi-winner voting.

Several axioms of representation for multi-winner elections are known from the literature;
here we recall two of them.

\begin{definition}[JR~\citep{aziz2017justified}]
A committee $S$ satisfies \emph{JR} if there exists no set of voters $V'\subseteq V$ with
$|V'| \geq n / k$, such that $|(\cap_{i \in V'} A_i)| \geq 1$ and $|((\cup_{i \in V'} A_i) \cap S)| = 0$. 
\end{definition}

\begin{definition}[PJR~\citep{sanchez2017proportional}]
A committee $S$ satisfies \emph{PJR} if for all $\ell\in [1, k]$ there does not exists a set of voters $V'\subseteq V$ with $|V'|\ge \ell \, n/k$
such that
$|\big(\cap_{i\in V'} A_i\big)| \ge \ell$ but 
$|\big( (\cup_{i\in V'} A_i)\cap S\big)| < \ell$.
\end{definition}

The intuition for JR is that a group of at least $n / k$ voters which agree on at least one candidate
shall not be completely ignored when forming the committee.
PJR generalizes JR by considering groups of $\ell n / k$ which agree on at least $\ell$ candidates
and requires such groups to be represented appropriately.

\section{Proportionality Axioms for PB}\label{section:axioms}

In this main section we present a series of proportionality axioms that are inspired by justified representation axioms in approval-based multi-winner voting.
The relations between these axioms are discussed in Section~\ref{section:relations} and are pictorially represented in Figure~\ref{fig:relations}.
For each axiom,
we study whether a budget satisfying it is guaranteed to exist,
the complexity of testing whether a specific budget satisfies it,
and the complexity of computing a budget satisfying it.

Informally, 
the general principle of our generalization of \emph{justified representation} (JR) to PB
is that a cohesive group of size $\ge n/L$ should control at least one unit of the budget
while our generalization of PJR require that,
for every $\ell \in [1, L]$,
a cohesive group of size $\ge \ell \cdot n/L$ should control at least $\ell$ units of the budget.

\begin{table*}[t]
\centering
\scalebox{0.65}{
\begin{tabular}{c | c c c | c c | c c c}
  \toprule
  & \multicolumn{3}{c|}{Existence guaranteed?} & \multicolumn{2}{c|}{Computational complexity of testing} & \multicolumn{3}{c}{Computational complexity of computing} \\
  & W & W+EX & L & W & L & W & W+EX & L \\
  \midrule	
  BJR &
    yes & yes & yes &
    P (Prop.~\ref{testing bjr}) & P (Prop.~\ref{testing bjr}) & 
    P ($W = \emptyset$) & 
    P (Prop.~\ref{proposition bjrl p}) &
    P (Prop.~\ref{proposition bjrl p}) \\
  Strong-BJR &
    yes & no & no &
    P (Prop.~\ref{testing bjr}) & P (Prop.~\ref{testing bjr}) &
  	P ($W = \emptyset$) & 
	NP-h (Prop.~\ref{set cover}) & 
	NP-h (Prop.~\ref{set cover}) \\ \midrule
  Local-BPJR &
    yes & yes & yes &
    Co-NP-c & Co-NP-c & 
	P ($W = \emptyset$) & 
	P (Prop.~\ref{prop:algPhragmenimpliesLocal BPJRL}; GPseq) &
	P (Prop.~\ref{prop:algPhragmenimpliesLocal BPJRL}; GPseq) \\
  BPJR &
    yes & yes & yes &
    Co-NP-c & Co-NP-c & 
    P ($W = \emptyset$) & 
    Open &
    NP-h (Prop.~\ref{BPJR-L NP-hard}) \\
  Strong-BPJR &
    yes & no & no &
    Co-NP-c & Co-NP-c & 
  	P ($W = \emptyset$) & 
	NP-h (Prop.~\ref{set cover}) & 
	NP-h (Prop.~\ref{set cover}) \\
  \bottomrule
\end{tabular}}
\caption{Summary of our results.
For each proportionality axiom $\mathcal{PA}$ we state
(1) whether a budget satisfying $\mathcal{PA}$ is guaranteed to exist;
(2) what is the computational complexity of testing whether a specific budget satisfies $\mathcal{PA}$;
and
(3) what is the computational complexity of computing a budget satisfying $\mathcal{PA}$.
(NP-h denotes NP-hardness, Co-NP-c denotes Co-NP-completeness, EX denotes exhaustiveness,
and GPseq denotes that our generalization of Phragmen's sequential rule satisfies the axiom.
The co-NP-c results follow from the fact that testing PJR is co-NP-complete in approval-based multi-winner voting~\citep{AzHu16a}.)}
\label{table:summary}
\end{table*}

\subsection{Strong-BJR and Strong-BPJR}

Our first formal extensions of JR and PJR to PB,
termed Strong-BJR and Strong-BPJR,
are given next.

\begin{definition}[Strong-BJR-L]
For a budget limit $L$, a budget $W$ satisfies \emph{Strong-BJR-L} if there exists no set of voters $V'\subseteq V$ with
$|V'| \geq n / L$, such that $w\big(\cap_{i \in V'} A_i\big) \geq 1$ but $w\big((\cup_{i \in V'} A_i) \cap W\big) =0$. 
\end{definition}

\begin{definition}[Strong-BPJR-L]
For a budget limit $L$, a budget $W$ satisfies \emph{Strong-BPJR-L} if for all $\ell\in [1, L]$ 
there does not exist a set of voters $V'\subseteq V$ with $|V'|\ge \ell \, n/L$, such that $w\big(\cap_{i\in V'} A_i\big)\ge \ell$ but 
$w\big( (\cup_{i\in V'} A_i)\cap W\big)<\ell$.
\end{definition}

The definitions above capture the idea that cohesive groups of at least $n/L$ voters should control (at least) one unit of budget. 
The Strong-BJR-L definition is `extreme' in the sense that a group of $n/L$ voters and a group of $2 \cdot n/L$ voters are treated the same;
that is, they are only guaranteed a non-zero amount of budget spent on representing them.
Put differently,
if voters simply desire some item from their approval ballot $A_i$ to be included in $W$ but do not care about the amount of budget spent on it,
then Strong-BJR-L might be sufficient - however, a more natural approach is to scale the amount of budget spent according to the size of the voter groups.
Indeed,
Strong-BPJR-L captures this more natural idea that a group of $n/L$ voters should not be treated the same as a group of $2 \cdot n/L$ voters -
the larger group should control (at least) twice as much budget than the smaller group;
in other words, the amount of budget controlled by a group of cohesive voters is proportional to their size.

\begin{remark}
Strong-BJR-L (Strong-BPJR-L) indeed generalize the definitions of JR (PJR) for multi-winner voting:
  Strong-BJR-L (Strong-BPJR-L) collapses to JR (PJR) when all items cost $w_i \equiv 1$.
\end{remark}

%
%

Strong-BJR-L and Strong-BPJR-L are appealing axioms in terms of their proportionality requiremtn.
Unfortunately, even Strong-BJR-L budgets are not guaranteed to exist
(notice that Strong-BPJR-L implies Strong-BJR-L; to see this, take $\ell = 1$).

\begin{example}
Consider a PB scenario with items $C=\{c_1, c_2, c_3\}$, costs $w(c_1)=w(c_2)=2$ and $w(c_3)=1$,
budget limit $L = 3$, and voters $V=\{1, \ldots, 4\}$ with the following ballots.
\begin{align*}
  A_1=A_2&=\{c_1\} \\
  A_3=A_4&=\{c_2\}.
\end{align*}

Then,
to satisfy Strong-PJR-L,
the group of voters $V' = \{A_1, A_2\}$, being of size $2 \geq n/L = 4/3$, deserves at least one unit of budget; similarly for the group of voters $V'' = \{A_3, A_4\}$.
Then, to satisfy the group $V_1$ ($V_2$) we shall include $c_1$ ($c_2$) in the budget.
The problem is, however, that $w(c_1) + w(c_2) = 4 > L = 3$; in other words, with the given budget limit, we cannot afford it.
This means that no budget satisfying Strong-BJR-L exists (and thus also Strong-BPJR-L). 
\end{example}

The key issue illustrated in the example above is that each set of cohesive voters,
namely $V'=\{A_1, A_2\}$ and $V''=\{A_3, A_4\}$,
supports only expensive items relative to their size: $V'$ ($V''$) supports $c_1$ ($c_2)$ whose cost is $2$ ($2$).
Intuitively,
since both groups constitute exactly $1/2$ of the electorate but at the same time each `demands' an item of cost $2$,
no feasible budget can satisfy the `demands' of both groups.
In Section~\ref{section:axioms not strong} we describe weaker versions of Strong-BJR-L and Strong-BPJR-L which do not suffer from this issue.

We end this section by considering the complexity of computing Strong-BJR-L and Strong-BPJR-L budgets.

\begin{proposition}\label{set cover}
  Computing Strong-BJR-L and Strong-BPJR-L budgets is NP-hard.
\end{proposition}

\begin{proof}
We reduce from the NP-hard Restricted-X3C problem~\cite{GJ79} in which,
given sets $S_1, \ldots, S_{3m}$ over a universe $x_1, \ldots, x_{3m}$,
where each set $S_j$ contains $3$ elements and each element is contained in $3$ sets,
the task is to exactly cover the universe with $m$ sets.

Given an instance of Restricted-X3C we construct a PB scenario as follows:
  For each set $S_j$ ($j \in [3m]$) we construct an item $c_j$ of cost $3$;
  For each element $x_i$ ($i \in [3m]$) we have a voter $v_i$ approving all items $c_j$ which correspond to sets $S_j$ which contain $x_i$
  (i.e., $v_i = \{c_j \in C : x_i \in S_j\}$).
  We set the budget limit $L$ to $3m$.
This finishes the construction.

Given $m$ sets which cover the universe,
we select exactly those (each cost $3$ so we respect the budget limit).
Then, each voter gets some representative, and Strong-BJR-L is satisfied.
Strong-BPJR-L is also satisfied, as each group of $\ell n / L = \ell$ voters
is covered by at least $\ceil{\ell / 3}$ items, of total cost at least $\ell$.

For the other direction, notice that, since $L = n$, it follows that even a Strong-BJR-L budget shall make sure that each voter gets represented by at least one item.
Recalling the budget limit, it means that a Strong-BJR-L budget shall correspond to an exact cover.
\end{proof}

\subsection{BJR and BPJR}\label{section:axioms not strong}

To remedy the possible non-existence of Strong-BJR-L and Strong-BPJR-L budgets described above,
below we introduce a weakening of these axioms which address this issue;
informally speaking, this is done by requiring each cohesive group of voters to (unanimously) support some item which is sufficiently cheap.  

\begin{definition}[BJR-L]
A budget $W$ satisfies \emph{BJR-L} if there exists no set of voters $V'\subseteq V$ with
$|V'|\ge n/{L}$ such that $w\big(\cap_{i\in V'} A_i\big)\ge 1$,
$w\big((\cup_{i\in V'} A_i) \cap W\big) = 0$,
and there exists some $c\in \cap_{i\in V'} A_i$ with $w(c)=1$.
\end{definition}

\begin{remark}\label{remark:weak}
The additional requirement for cohesive groups of voters in the BJR definition is very restrictive.
For example, if there is only one item of cost one, then any budget $W$ containing this item will satisfy JR.
Thus,
the definition of BJR is only meaningful when there are multiple items of cost one.
This observation suggests that the JR budget definition is useful as a `minimal' requirement of representation
but
leaves a lot to be desired in terms of representation. 
\end{remark}

In a similar manner to the definition of BJR-L,
below we introduce the axiom of BPJR-L.
Here,
however,
our requirement allows for greater flexibility in the additional requirement of cohesive voters supporting a `sufficiently cheap' item.
In particular, we allow for bundles of items and let the `sufficiently cheap' criteria grow proportionally with the size of the voters group. 

\begin{definition}[BPJR-L]
A budget $W$ satisfies \emph{BPJR-L} if for all $\ell\in [1, {L}]$
there exists no set of voters $V'\subseteq V$ with $|V'|\ge \ell \, n/{L}$ such that $w\big(\cap_{i\in V'} A_i\big)\ge \ell$
and
$w\big( (\cup_{i\in V'} A_i)\cap W\big)< \max\Big\{w(C')\, :\, C'\subseteq \cap_{i\in V'}A_i \, \text{ and }\, w(C')\le |{V'|L/n}\Big\}$.
\end{definition}

The idea of the definition above is that every group of (at least) $\ell \cdot n/L$ voters should control (at least) $\ell$ units of the budget.
Technically,
however,
for this to be satisfied we need a bundle of items with `low enough' cost $\le \ell$.
We then require that at least the cost of the most expensive bundle which is also `low enough' is used in representing this group of $\ell \cdot n/L$ voters.
Note that indeed BPJR-L generalizes BJR-L; specifically, BPJR-L boils down to BJR-L when $\ell = 1$.

It turns out that exhaustive BPJR-L budgets always exist.
(While the procedure described in the next proof is super-polynomial,
afterwards we show that this is not a coincidence,
by showing the problem of computing a budget satisfying BPJR-L is NP-hard.)

\begin{proposition}\label{proposition:exhaustive pjr b always exist}
  For any given budget limit $L$ there always exists a feasible budget $W$ which is exhaustive and satisfies BPJR-L.
\end{proposition}

\begin{proof}
We describe an algorithm which produces feasible, exhaustive budgets, which satisfy BPJR-L.
The algorithm proceeds as follows.
We iterate over $\ell'$, where initially $\ell' = L$ and $\ell'$ can only decrease, until $\ell' = 1$.
Let $A' = A$ (initially, considering all voters; during the course of the algorithm, we will ``discard'' voters as we take care for them),
and let $W = \emptyset$ (initially, no item is budgeted).

In each iteration,
first check whether $w(W) + \ell' \leq L$;
if this is not the case, then decrease $\ell'$ by one and continue to the next iteration.
Otherwise,
let $$C^*:=\{C' \subseteq C\, :\, w(c) = \ell'\}.$$
If $C^* = \emptyset$, then decrease $\ell$ by one and continue to the next iteration.
Otherwise,
for each $C' \in C^*$, let
$$A(C') = \{ i \in A' \, : \, C' \in A_i \}$$
denote the voters from $A'$ which approve all items of $C'$.
Select any $C' \in C^*$ with maximal size of $|A(C')|$
and check whether $|A(C')| \geq \ell' \cdot n / L$;
If this is not the case, then decrease $\ell'$ by one and continue to the next iteration.
Otherwise,
set $W \mapsto W \cup C'$
and
redefine $A' \mapsto A' \setminus A(C')$.
Leave $\ell'$ as is and continue to the next iteration.
The algorithm halts whenever $\ell' = 0$,
in which case,
if $W$ is non-exhaustive, then we arbitrarily add items to it until it does,
while keeping it feasible.

Next we prove the algorithm's correctness.
Let $W$ be the output from the above algorithm and for the purpose of a contradiction suppose that BPJR-L is not satisfied. That is, there exists an $\ell\in [1, L]$ and a set of voters $V'$ with $|V'|\ge \ell n/L$ and $w(\cap_{i\in V'} A_i)\ge \ell$ such that there exists $C'\subseteq \cap_{i\in V'} A_i$ with $w(C')\le \ell$ such that 
$$w\big( (\cup_{i\in V'} A_i)\cap W\big)< W(C').$$
Furthermore, let $\ell$ be the smallest number such that the above holds. 

Now suppose that $w(W)=L$. Then every unit spent in the budget represents at least $n/L$ additional voters, and so the total number of adequately represented voters is at least
$$w(W) \cdot n/L=n,$$
thus, there cannot be a group of $|V'|$ voters which are not represented by at least $|V'|/(\frac{n}{L})\ge \ell$ units of budget; thus, a contradiction.

Now suppose that $w(W)<L$. Then, every unit spent in the budget represents at least $n/L$ additional voters, and so the total number of adequately represented voters is at least
$$w(W) \cdot n/L<n.$$
Denote this set of voters by $R\subseteq V$ and note that any subset $V''\subseteq R$ is adequately represented in the budget with at least $|V''|/(\frac{n}{L})$ units of expenditure. 
Noting that there exists a group $V'$ of size $|V'|\ge \ell n/L$ who are inadequately represented -- recalling that $\ell$ is the smallest number such that this holds --- it must be that the group was not represented due to the budget limit and thus
$$w(W)+\ell>L.$$
Now if there exists a (smallest) group $V'$ of inadequately represented voters it must be that 
$$V'\subseteq V\backslash R$$
and so 
\begin{align*}
|V'|&\le n-|R|\\
&\le n-\frac{w(W) n}{L}\\
&=\frac{n}{L}\Big(L-w(W)\Big)\\
&<\ell \frac{n}{L}&&\text{since $w(W)+\ell>L$,}\\
&\le |V'|,
\end{align*}
thus no such set $V'$ can exist and we have derived a contradiction. We conclude that BPJR-L must be satisfied.
\end{proof}


The next result explains why the algorithm presented in the proof above is not polynomial-time.

\begin{proposition}\label{BPJR-L NP-hard}
  Computing a BPJR-L budget is NP-hard.
\end{proposition}

\begin{proof}
%
We describe a reduction from the NP-hard problem Partition~\cite{GJ79} which,
given integers $x_1, \ldots, x_m$ whose sum is $2B$,
asks to decide whether there is a subset of them which sums to $B$.

Given an instance of Partition with integers $x_1, \ldots, x_m$ of sum $2B$,
we construct a PB scenario for which BPJR-L budget exists if and only if the
there is a subset of those integers whose sum equals to $B$.
Specifically,
we construct items $c_1, \ldots, c_m$, where item $c_j$ ($j \in [m]$) costs $x_j$ and have $1$ voter approving all of them;
we set the budget limit $L$ to $B$. This finishes the construction.

Given a solution to the Partition instance, consisting of a subset $X' \subseteq X$ with $\sum_{x' \in X'} x' = B$,
we construct a budget $W = X'$ which costs $B = L$ and thus is exhaustive and satisfies BPJR-L.
For the other direction, notice that, according to BPJR-L, the single voter deserves the whole budget,
thus a budget satisfying BPJR-L shall correspond to a solution to the Partition instance.
\end{proof}

Computing BJR-L budgets, however, can be done in polynomial time.
The proof of the next proposition is by a greedy algorithm, somehow resembling Approval-based greedy Chamberlin--Courant.

\begin{proposition}\label{proposition bjrl p}
  There is a polynomial-time algorithm which computes an exhaustive budget satisfying BJR-L. 
\end{proposition}

\begin{proof}
If $L=0$, 
then $W=\emptyset$ satisfies the proposition.
Let $L>0$ be a budget limit and define the set of cheapest items:
  $$C^*:=\{c\in C\, :\, w(c)=1\}.$$
If $|C^*|\le L$,
then any budget $W$ such that $C^*\subseteq W$ is feasible and satisfies BJR-W.
To satisfy the exhaustiveness property we add items $c\notin W$ to the budget until it is exhaustive.

If $|C^*|>L$,
then we continue according to the following procedure.
Let $A'=A$, $W=\emptyset$,
and let 
$s(c, A')=|\{i\in A'\, :\, c\in A_i\}|$
denote the approval score of item $c$ with respect to the ballot profile $A'$.
Select any $c\in C^*$ with maximal score $s(c, A')$;
then, set $W\mapsto W\cup\{c\}$,
remove all approval ballots with $c \in A_i$ from $A'$,
and redefine $C^*\mapsto C^*-\{c\}$.
Repeat this process until $w(W)\ge \floor{L}$ or until $C^*=\emptyset$. 

Note that at each stage where an item $c\in C^*$ is added to $W$,
a group of $s(c, A')$ unrepresented voters become represented.
Also note that at each stage $s(c, A')$ is weakly decreasing;
that is, we remove voters who have been represented and so the approval score of any item can never increase. 

In the first case, the algorithm terminates with $W$ such that $w(W)\ge \floor{L}$,
thus the exhaustiveness and feasibility properties are satisfied.
Now,
for the purpose of a contradiction,
suppose that BJR-W-L is not satisfied.
Thus, there exists an item $c\in C^*$ with $s(c, A')\ge n/L$,
for $A'$ at the algorithm's termination.
But since $s(c, A')$ is weakly decreasing at each stage and the item $\tilde{c}\in C^*$ with maximal score is added each time, but $c\in C^*$ was never elected, it must be that every one of the $|W|=w(W)$ items added to $W$ represented at least $|W| \cdot n/L$ distinct voters.

Notice, however, that as
\begin{align*} 
|W|\cdot n/L &\ge \floor{L}n/L
> (L-1) n/L =n-n/L,
\end{align*}

%
it would mean that strictly more than $n-n/L$ voters were represented -- meaning that there can no exists a set $V'$ of unrepresented voters with $|V'|\ge n/L$; this gives a contradiction. 

For the second case, note that there is no cheapest item supported by any voter - let alone a cohesive group of $\ge n/L$ voters.
Thus, the budget $W$ satisfies BJR-L.
To ensure that the exhaustiveness property is satisfied we (arbitrarily) add items $c\notin W$ until the budget is exhaustive (but still remains feasible). 
\end{proof}

We end this section by considering the complexity of testing whether a given budget satisfies BJR-L (and Strong-BJR-L).

\begin{proposition}\label{testing bjr}
  There is a polynomial-time algorithm to test whether a given budget satisfies BJR-L and Strong-BJR-L.
\end{proposition}

\begin{proof}
Given a PB scenario with items $C$, voters $V$, and budget limit $L$,
and a budget $W$, the task is to decide whether $W$ satisfies BJR-L or Strong-BJR-L.
We proceed by describing a procedure for BJR-L and mention how it shall be modified for Strong-BJR-L.

First, we find all voters which are not represented at all;
formally, let $V'' := \{v \in V : V \cap W = \emptyset\}$.
Next,
for each $c \in C$ (or, for Strong-BJR-L, for each $c \in \{c \in C : w(c) = 1\}$),
  consider the voters in $V''$ which approve $c$; formally, let $V''_C = \{v'' \in V'' : c \in v'\}$.
  Then, if $|V''_c| \geq n/L$,
  reject.
If reached the end, accept.
\end{proof}

\subsection{Local-BPJR}

As we are interested in efficient budgeting methods which output budgets satisfying certain forms of proportional representation,
the computational hardness result of the last section is somewhat disappointing,
as it presumably rules out the possibility of efficient methods which compute exhaustive BPJR-L budgets.
Here we consider weaker versions of these concepts and then we describe an efficient budgeting method which computes budgets satisfying it.
\begin{definition}[Local-BPJR-L]
A budget $W$ satisfies \emph{Local-BPJR-L} if for all $\ell\in [1, {L}]$ there exists no set of voters $V'\subseteq V$ such that $W'=(\cup_{i\in V'} A_i)\cap W$, $|V'|\ge \ell \, n/{L}$ and there exists some $W''\supset W'$ such that
$$W''\in \max\Big\{w(C')\, :\, C'\subseteq \cap_{i\in V'}A_i \, \text{ and }\, w(C')\le \ell\Big\}.$$
\end{definition}


\citet{Jans16a} reports on several interesting rules, developed by Phragmen, which were designed to achieve proportionality axioms in multi-winner voting. 
\citet{BFJL16a} proved that one of these rule, which they referred to as \emph{Phragmen's sequential rule}, computes a committee that satisfies PJR
(this refers to the proportionality axiom for multi-winner voting).
Here,
we generalize Phragmen's sequential rule~\citep{BFJL16a,Jans16a} to the case of PB.

\begin{algorithm}[t]
  \begin{algorithmic}
    \REQUIRE $(N, C, w, L)$ \COMMENT{resp.: voters, items, cost function, limit}
    \ENSURE $W$ \COMMENT{budget}
  \end{algorithmic}
  \begin{algorithmic}[1]
    \STATE $W\longleftarrow \emptyset$
    \WHILE{$C'' = \{c \notin W : w(W) + w(c) \leq L \land \exists i \in N : c \in A_i\} \neq \emptyset$}\label{line:three}
      \STATE Let $C^* = \{c' : c' \in \argmin_{c' \in C''} s_{c'}\}$ where: \COMMENT{argmin set}

\STATE \quad $x_{c,i} \geq 0$ \qquad ($ \forall c \in C$, $\forall i \in N$)
\STATE \quad $x_{c,i} = 0$ \qquad ($\forall c \in C$, $\forall i \in N$ such that $c \notin A_i$)
\STATE \quad $\sum_{i\in N} x_{c,i} = w(c)$ \qquad ($\forall c \in W \cup \{c'\}$)
\STATE \quad $\sum_{i\in N} x_{c,i} = 0$ \qquad ($\forall c \notin W\cup \{c'\}$)
\STATE \quad $x_i = \sum_{c\in C} x_{ci}$ ($\forall i \in N$)
\STATE \quad $s_{c'} \geq x_i$ \qquad ($\forall i \in N$)
\STATE Let $c^* \in C^*$ \COMMENT{break ties arbitrarily}
      \STATE $W \longleftarrow W \cup \{c^*\}$
    \ENDWHILE
    \RETURN $W$
	\end{algorithmic}
	\caption{Generalized Phragmen's sequential rule for PB (GPseq).}
	\label{algorithm:phragmen}
\end{algorithm}

Our generalized rule is referred to as GPseq (Generalized Phragmen's sequential rule)
and proceeds as follows.
Items are iteratively added until no item can be added without exceeding the budget limit.
An item that is added is required to spread its cost among voters who approve it.
When an item is considered to be added to the set of selected items,
we check what will be the maximum cost received by a voter.
We select the item that minimizes the maximum cost received by voters (we discuss tie breaking later).
The rule is also shown in Algorithm~\ref{algorithm:phragmen}.

\begin{remark}
When distributing the cost of the current item, we are allowed to redistribute the cost of the already-chosen items; this can be directly implemented using a linear program (as in Algorithm~\ref{algorithm:phragmen}), but also shown to be polynomial-time solvable using a combinatorial argument, as shown by Brill et al.~\cite{BFJL16a}.
\end{remark}

\begin{remark}
Currently, Line~\ref{line:three} in Algorithm~\ref{algorithm:phragmen} does not consider items $c$ which are not approved by any voter.
As can be seen from the proof of Proposition~\ref{prop:algPhragmenimpliesLocal BPJRL},
leaving the algorithm as is still results in it satisfying Local-BPJR.
Perhaps unsatisfactory,
however,
the resulting budgets might be non-exhaustive due to some items $c$ which are not approved by any voter (but fits within the budget limit).

One possible fix is by including a post-processing phase,
specifically looking for such items at the end of the algorithm and adding them exhaustively.
In situations where not approving an item  simply means not caring for him,
it might makes sense;
however, if not approving an item actually means disapproving him,
then it might reduce the social welfare.
\end{remark}

Before we prove that GPseq satisfies Local-BPJR-L,
the next example shows that, unfortunately, it does not satisfy BPJR-L.

\begin{example}\label{proposition:phragmen seq do not satisfy pjr b}
Consider the following instance with $C=\{a, b, c,d\}$ and $w(a)=2, w(b)=w(c)=1.5$ and $w(d)=1$.
Let there be 6 voters with approvals
\begin{align*}
A_1=\ldots=A_4&=\{a,b\}~~~\quad
A_5=A_6=\{c\}.
\end{align*}
If GPseq is run with $L=3$, then in the first iteration $W_1=\{b\}$, and in the second we add item $c$. Thus the computed budget is:
$$W=\{b, c\} \qquad \text{ and } \qquad w(W)=3.$$
This budget does not satisfy BPJR-W since the group of voters $V'=\{1, \ldots, 4\}$ is of size $\ge 2 \cdot n/w(W)=4$ and unanimously support a bundle $\{a\}$ with cost $2$ but they were only represented by a bundle of cost $1.5$ (i.e., item $b$).
\end{example}

\begin{proposition}\label{prop:algPhragmenimpliesLocal BPJRL}
	GPseq satisfies Local-BPJR-L.
\end{proposition}

\begin{proof}
Assume, towards a contradiction, that there is a PB scenario where GPseq outputs a budget $W$ which violates Local-BPJR-L.
By definition, this means that there is a number $\ell\in [1, L]$, a set of voters $|V'|\ge \ell n/L$,
and some $W''\supset W':=(\cup_{i\in V'} A_i)\cap W$ such that 
$$W''\in \max\Big\{w(C')\, :\, C'\subseteq \cap_{i\in V'}A_i \, \text{ and }\, w(C')\le \ell\Big\}.$$

First, observe that it must be the case that $w(W')<\ell$ (since otherwise no such $W''$ could be feasible).
As $W' \subset W''$, there must be some $c^* \in W'' \setminus W'$,
and, as $W' =(\cup_{i\in V'} A_i)\cap W$, it must be that $c^* \notin W$.
Further,
since $W' \cup \{c^*\} \subseteq W''$, the following holds:
\begin{align}\label{equation1: algPhragmenimpliesLocal BPJRL}
  w(W')+w(c^*)\le w(W'')\le  \ell.
\end{align}

Recall that GPseq works in iterations, where in each iteration another item is added to the intermediate budget $W$.
Further, before the first iteration it is possible to add $c^*$ to the partial budget, since the partial budget is empty;
while after the last iteration it is surely not possible anymore, as if it was so, then GPseq would not terminate.
The proof would now follow by considering the iteration at which the corresponding `switch' had happened,
and would show that the maximum voter spread could have been smaller if $c^*$ would have been chosen instead of the other item which was chosen;
this would contradict the way by which GPseq works.

To be more formal,
the following notation is helpful.
Denote the intermediate budget at the completion of the $j$th iteration of GPseq by $W^j$;
and by $x_i^j$ and $s_j^*$, the spread of voter $i$ and the maximum voter spread (respectively; specifically, $s_j^* = \max_{i \in V} x_i^j$)
at the completion of that iteration.

Supplied with the above notation,
let $j$ be the index of the first iteration for which the following hold:
\begin{align}
\label{equation: jearlystage2}
w(W^{j-1}) + w(c^*)&\le L,\\
\label{equation: jearlystage1}
w(W^j) + w(c^*)&>L.
\end{align}
Recall that such $j$ must exist as otherwise GPseq would not have terminated.

Next we compute a lower bound on the maximum voter spread in $W^j$.
As $W' = (\cup_{i \in V'} A_i) \cap W$ it follows that the total weight of the items in $W_j \setminus (W' \cap W^j)$ is spread over at most
$n - |V'|$ voters (those corresponding to voters not in $V'$).
Then, from averaging, if follows that there must be at least one voter $k\in V\backslash V'$ for which the following hold:

\begin{align*}
x_k^j&\ge \frac{w(W^j)-w(W'\cap W^j)}{n-|V'|}\\
&>\frac{L-w(c^*)-w(W'\cap W^j)}{n-|V'|}&&\text{ by (\ref{equation: jearlystage1})}\\
&\ge \frac{L-\ell}{n-|V'|}&&\text{by (\ref{equation1: algPhragmenimpliesLocal BPJRL})}\\
&\ge \frac{L-\ell}{\frac{n}{L}(L-\ell)}&&\text{since $|V'|\ge \ell\frac{n}{L}$,}\\
&= \frac{L}{n}.
\end{align*}
As $s_j^* = \max_{i \in V} x_i^j$,
it follows that $$s_j^*> \frac{L}{n}.$$


%


Next we compute an upper bound on the maximum voter spread if,
instead of adding the item added in the $j$th iteration,
we would add $c^*$ at that iteration. Thus, the budget at the completion of the $j$th iteration will be ${W^j}'=W^{j-1}\cup\{c^*\}$.
Let us denote the alternative values of $x_i^j$ and $s_j^*$ by ${x_i^j}'$ and ${s_j^*}'$, respectively.

Recall that GPseq considers minimizing the maximum voter spread,
thus to compute an upper bound it is sufficient to consider one way of spreading $c^*$'s weight;
we will consider spreading it evenly among the voters in $V'$.
That way, the spread of voters $k \in V\backslash V'$ remains unchanged to the previous iteration $j-1$;
i.e., for $k \in V\backslash V'$ it holds that:
\begin{align}\label{equation: equalV'}
{x_k^j}'=x_{k}^{j-1}.
\end{align}
Let 
\begin{align}\label{equation: maxofothervotersj-1}
\tilde{s}_{j-1}^*=\max_{k\notin V'} x_k^{j-1}=\max_{k\notin V'} {x_k^j}',
\end{align}
this denotes the maximum spread among voters in $V\backslash V'$ at the $(j-1)$th iteration (or equivalently the alternative maximum spread of such voters in the $j$th iteration). Note that some weight of items in $W'\cap W^{j-1}$ may be distributed to voters in $V\backslash V'$ at this $(j-1)$th iteration.

Further, by redistributing the weight at the $(j-1)$th iteration
$$\sum_{i\in V'} x_i^{j-1}\le w(W'\cap W^{j-1})$$
evenly among voters in $V'$ we can maintain equality in (\ref{equation: equalV'}). We then construct the alternative value of spreads for voters in $V'$ by evenly distributed the weight of item $c^*$ as well. Thus, for all $i\in V'$
\begin{align*}
{x_i^j}'&=\frac{\sum_{i\in V'} x_i^{j-1}}{|V'|}+\frac{w(c^*)}{|V'|}\\
&\le \frac{w(W'\cap W^{j-1})}{|V'|}+\frac{w(c^*)}{|V'|}\\
&\le \frac{\ell}{|V'|}&&\text{by (\ref{equation1: algPhragmenimpliesLocal BPJRL})}\\
&\le \frac{L}{n}.
\end{align*}
%
%

%
Recall (\ref{equation: maxofothervotersj-1}), it follows that under this spread,
$${s_j^*}'=\max\Big(\tilde{s}_{j-1}^*,\max_{i\in V'}({x_i^j} ')\Big).$$

Next we consider two cases.
First, if $\tilde{s}_{j-1}^* \le \frac{L}{n}$, then we are done as this would imply that 
$${s_j^*}'\le \frac{L}{n}<s_j^*,$$
which contradicts the fact the GPseq minimizes $s_j^*$ at each iteration.

For the second case,
suppose that $\tilde{s}_{j-1}^*>\frac{L}{n}$ and so 
$$ s_{j-1}^*>\frac{L}{n}.$$
This would imply that at some earlier iteration, say $t<j$, we have 
$$s_{t-1}^*\le \frac{L}{n} \quad \text{ and } s_t^*>\frac{L}{n}$$
(as initially no weight is spread on the voters).
Further, as $c^* \notin W$ and $W^{t-1} \subseteq W$ and $W^t \subseteq W$,
for the intermediate budgets $W^{t-1}$ and $W^t$ it holds that
$$w(W^{t-1})<w(W^t)\le w(W^{j-1})\le L-w(c^*).$$
This means that adding the item $c^*$ does not cause us to exceed the budget limit $L$ since $t<j$ and $j$ is the earliest stage such that (\ref{equation: jearlystage2}) and (\ref{equation: jearlystage1}) are satisfied.
Thus, if at stage $t$ instead item $c^*$ was added to the intermediate budget $W^t$ then by spreading the additional weight $w(c^*)$ among agents in $V'$ we can attain the following alternative maximum voter spread:
$${s_t^*}' \le\frac{L}{n}<s_t^*.$$
This contradicts the fact that GPseq minimizes $s_j^*$ at each iteration. 

Overall,
we conclude that no such set $V'$ can exist and thus Local-BPJR-L is satisfied. 
\end{proof}

\paragraph{Tie breaking}

Notice that both Theorem~\ref{prop:algPhragmenimpliesLocal BPJRL} and Example~\ref{proposition:phragmen seq do not satisfy pjr b}
are
oblivious to the tie breaking used by GPseq.
Arguably,
its tie-breaking can result in some not-intuitive behavior,
as the following example shows.
Let $C=\{c_1, c_2\}$ with $w(c_1)=1$ and $w(c_2)=2$ and let there be voters $V=\{1,\ldots, 6\}$ with the following ballots:
\begin{align*}
  A_1=\cdots =A_4&=\{c_2\}~~~ \quad
  A_5=A_6=\{c_1\}.
\end{align*}
Then,
if we have a budget limit $L=2$, then the items to be chosen for Phragmen-budget will be either $c_1$ or $c_2$.
Now,
if we favor cheaper items, or simply choose $c_1$ arbitrarily,
then we will `satisfy' only 2 voters;
whilst, perhaps a more intuitive outcome would be $W'=\{c_2\}$ which is of a cost $w(W)=L=2$ and satisfies 4 voters.
Both budgets, however, satisfy BPJR-L in this case.



\subsection{W-variants and Relations between Axioms}
\label{section:relations}

So far our axioms of proportionality depended on an external budget limit $L$.
Indeed, being based on $L$, such definitions are easily communicated, as usually $L$ is known before hand;
further, each group of voters can easily compute and appreciate the fraction of $L$ which they can claim for themselves.

There is, however, some merit in being oblivious to $L$, by considering proportionality axioms which are oblivious to $L$, and are properties of the budget itself (with respect to the electorate, of course). This is possible through what we refer to as W-variants.
Specifically, instead of considering groups of voters of $\ell \cdot n / L$ (which, as our definitions above state, deserve $\ell$ units of the budget), in our W-variants we concentrate not on the external budget limit~$L$ but on the actual total cost of the budget $w(W)$, and consider groups of voters of $\ell \cdot n / w(W)$, which deserve $\ell$ units of the budget.
We feel that these definitions, which are based on $w(W)$ are, mathematically speaking, more elegant,
as they are properties of the budget itself and are oblivious to the externally-imposed budget limit. 
Formally, we suggest the following definitions (which are analogous to their L-variants described in the previous sections).

\begin{definition}[Strong-BJR-W]
A budget $W$ satisfies \emph{Strong-BJR-W} if there exists no set of voters $V'\subseteq V$ with
$|V'| \geq n / w(W)$, such that $w\big(\cap_{i \in V'} A_i\big) \geq 1$ but $w\big((\cup_{i \in V'} A_i) \cap W\big) =0$. 
\end{definition}

\begin{definition}[Strong-BPJR-W]
A budget $W$ satisfies \emph{Strong-BPJR-W} if for all $\ell\in [1, w(W)]$ there does not exist a set of voters $V'\subseteq V$ with $|V'|\ge \ell \, n/w(W)$, such that $w\big(\cap_{i\in V'} A_i\big)\ge \ell$ but 
$w\big( (\cup_{i\in V'} A_i)\cap W\big)<\ell$.
\end{definition}

\begin{definition}[BJR-W]
A budget $W$ satisfies \emph{BJR-W} if there exists no set of voters $V'\subseteq V$ with
$|V'| \geq n / w(W)$ such that $w\Big(\cap_{i \in V'} A_i\Big) \geq 1$,
$w\Big((\cup_{i \in V'} A_i) \cap W\Big) =0$,
 and there exists some $c\in \cap_{i \in V'} A_i$ with $w(c) = 1$.
\end{definition}

\begin{definition}[BPJR-W]
A budget $W$ satisfies \emph{BPJR-W} if for all $\ell\in [1, w(W)]$
there exists no set of voters $V'\subseteq V$ with $|V'|\ge \ell \, n/w(W)$ such that $w\big(\cap_{i\in V'} A_i\big)\ge \ell$
and
$$w\big( (\cup_{i\in V'} A_i)\cap W\big)< \max\Big\{w(C')\, :\, C'\subseteq \cap_{i\in V'}A_i \, \text{ and }\, w(C')\le \ell\Big\}.$$
\end{definition}

\begin{definition}[Local-BPJR-W]
A budget $W$ satisfies \emph{Local-BPJR-W} if for all $\ell \in [1,L]$ there exists no set of voters $V'\subseteq V$ such that
$W'=(\cup_{i\in V'} A_i)\cap W$, $|V'|\ge \ell \, n/w(W)$ and there exists some $W''\supset W'$ such that
$$W''\in \max\Big\{w(C')\, :\, C'\subseteq \cap_{i\in V'}A_i \, \text{ and }\, w(C')\le \ell\Big\}.$$
\end{definition}

Notice that each L-variant implies its corresponding W-variant.

\begin{proposition}
  Any L-variant implies its W-variant.
\end{proposition}

\begin{proof}
In the L-variants $\ell \in [1, L]$ while in the W-variants $\ell \in [1, w(W)]$,
and $w(W) \leq L$ always holds.
Further, the sets $V'$ considered in the L-variants satisfy $|V'| \geq \ell n / L$
while those in the W-variants satisfy $|V'| \geq \ell n / w(W)$;
again, $w(W) \leq L$ holds,
it follows that $n/L \leq n/w(W)$.
Therefore, each L-variant considers all sets considered by its corresponding W-variant.
\end{proof}

Notice that W-variants always exist, as, for example, the empty budget $B = \emptyset$ satisfies them (as our voter set is always finite). Thus, it makes more sense to require also exhaustiveness from such budgets; for example, considering the complexity of computing exhaustive Local-BPJR-W budgets. Indeed, tractability of computing an L-variant implies tractability of computing an exhaustive W-variant.
On the other hand, exhaustive W-variants might be computationally easier to compute than their corresponding L-variants; while we know that computing exhaustive Strong-BJR-W budgets and exhaustive Strong-BPJR-W budgets is NP-hard, we could not modify the proof of hardness of computing BPJR-L budgets (Theorem~\ref{BPJR-L NP-hard}) to apply to exhaustive BPJR-W budgets as well. We conjecture, however, that indeed computing exhaustive BPJR-W is NP-hard as well.

\nimrodout{
The only separation we have now is the following.
Next we show that the problem of computing an exhaustive BPJR-W is NP-hard in the weak sense;
we conjecture that it is in fact NP-hard in the strong sense.
	
\begin{proposition}\label{weak NP}
  Computing an exhaustive BPJR-W budget is weakly NP-hard.
\end{proposition}

\begin{proof}
We describe a reduction from the weakly NP-hard problem Partition which, given integers $x_1, \ldots, x_m$ whose sum is $2B$, asks to decide whether there is a subset of them which sum to $B$.

We construct a PB scenario with items $c_1, \ldots, c_m$, where item $c_j$ ($j \in [m]$) costs $x_j$ and further $2B$ unit-cost items, $d_1, \ldots, d_{2B}$. We have $2$ voters: voter $v_1$ approving items $c_1, \ldots, c_m$ and voter $v_2$ approving items $d_1, \ldots, d_{2B}$. The budget limit $L$ is $2B$.
This finishes the construction.

Given a solution to the Partition instance, consisting of a subset $X' \subseteq X$ with $\sum_{x' \in X'} x' = B$, we construct a budget $W = X' \cup \{d_1, \ldots, d_B\}$. $W$ costs $2B = L$ and thus is exhaustive and satisfies BPJR-W as $w(W) = L$ and each voter out of the two can claim half of the total budget as her representative.

For the other direction, notice that any exhaustive budget satisfies $w(W) = L$ and thus each of the voters shall be represented by items of total cost at least $B$; this means that the representatives of voter $v_1$ correspond to a solution to the Partition instance.
\end{proof}
}


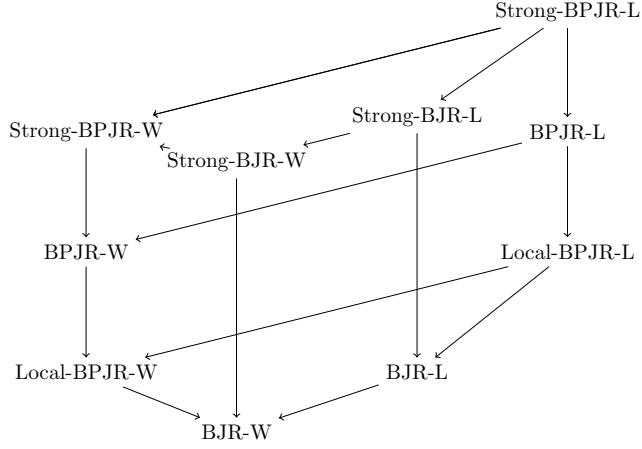
\begin{figure}
\begin{center}
\scalebox{0.8}{
\begin{tikzpicture}

\tikzstyle{pfeil}=[->,>=angle 60, shorten >=1pt,draw]
\tikzstyle{onlytext}=[]

\node[onlytext] (S-BPJR-L) at (4,1) {Strong-BPJR-L};
\node[onlytext] (BPJR-L) at (4,-1) {BPJR-L};
\node[onlytext] (Local-BPJR-L) at (4,-3) {Local-BPJR-L};
\node[onlytext] (BJR-L) at (1.5,-5) {BJR-L};
\node[onlytext] (S-BJR-L) at (1.5,-0.75) {Strong-BJR-L};
    
\node[onlytext] (S-BPJR) at (-4,-1) {Strong-BPJR-W};
\node[onlytext] (BPJR) at (-4,-3) {BPJR-W};
\node[onlytext] (Local-BPJR) at (-4,-5) {Local-BPJR-W};
\node[onlytext] (BJR) at (-1.5,-6) {BJR-W};
 \node[onlytext] (S-BJR) at (-1.5,-1.5) {Strong-BJR-W};
  
\draw[->] (S-BPJR-L) -- (BPJR-L) ;
\draw[->] (BPJR-L) -- (Local-BPJR-L) ;
\draw[->] (Local-BPJR-L) -- (BJR-L) ;
\draw[->] (S-BPJR-L) -- (S-BJR-L) ;

\draw[->] (S-BPJR-L) -- (S-BPJR) ;
\draw[->] (BPJR-L) -- (BPJR) ;
\draw[->] (Local-BPJR-L) -- (Local-BPJR) ;
\draw[->] (S-BPJR-L) -- (S-BPJR) ;
  
\draw[->] (S-BPJR) -- (BPJR) ;
\draw[->] (BPJR) -- (Local-BPJR) ;
\draw[->] (Local-BPJR) -- (BJR) ;
\draw[->] (S-BPJR) -- (S-BJR) ;

\draw[->] (S-BJR-L) -- (S-BJR) ;
\draw[->] (BJR-L) -- (BJR) ;

\draw[->] (S-BJR) -- (BJR) ;
\draw[->] (S-BJR-L) -- (BJR-L) ;

\end{tikzpicture}}
\end{center}
\caption{\label{fig:relations} Logical relations between proportionality concepts.
An arrow from (A) to (B) denotes that concept (A) implies concept (B).}
\end{figure}

We conclude the section by further discussing the logical relations between the axioms presented and discussed in the sections above. These relations are also pictorially represented in Figure~\ref{fig:relations}.
	
\begin{proposition}
  Strong-BPJR-(W/L) implies BPJR-(W/L) which implies Local-BPJR-(W/L) which implies BJR-(W/L).
\end{proposition}

\begin{proof}
The definition of BPJR-(W/L) includes a check on $V'$ which is not present in the definition of Strong-BPJR-(W/L), and thus is implied by it.
The definition of Local-BPJR-(W/L) considers specific sets $W''$ to represent certain groups, therefore is implied by BPJR-(W/L)
which consider more sets.

To show that Local-BPJR-(W/L) implies BJR-(W/L) we use the contrapositive by assuming that BJR-(W/L) is not satisfied.
This implies that there exists a group of voters $V'$ such that
$|V'|\ge n/w(W)$ (or $\ge n/L$),
$w(\cup_{i\in V'} A_i)\ge 1$,
and some $c\in \cap_{i\in V'} A_i$ with $w(c) = 1$ exists,
but nevertheless $w\Big((\cup_{i\in V'} A_i) \cap W\Big)=0$.
Recalling the definition of Local-BPJR-(W/L),
this implies that $W'=\emptyset$ and so Local-BPJR-(W/L) is not satisfied if there exists \emph{any} budget $W''$
(since every budget contains the empty set)
containing some subset of items $C'\subseteq \cap_{i\in V'} A_i$ and $w(C')\le \ell$ such that $w(W'')>w(W')$.
Such a budget always exists since a BJR-(W/L) budget always exists (Proposition~\ref{proposition:exhaustive pjr b always exist}),
and a budget satisfying BJR-(W/L) satisfies the conditions for $W''$ to fail $W$ as a Local BPJR-(W/L) budget.
\end{proof}

\section{Conclusions}

Participatory budgeting is an interesting and widely applicable setting, gaining growing attention from the research community and being more extensively deployed.
The axiomatic, normative study of methods of participatory budgeting is still lacking,
and issues of proportionality and representation are currently not well understood;
this is especially unfortunate as many times it is desirable to spend funds in a proportional way, taking into account issues of representativeness,
and not letting the majority control all the available budget.
Thus,
in this paper, we proposed several new, proportional representation axioms as well as efficient corresponding algorithms. Many of our results are summarized in Table~\ref{table:summary}. 

As we framed participatory budgeting as a generalization of multi-winner voting, our axioms and rules can also be viewed as interesting generalizations of work on multi-winner voting. Some of the interesting insights include the following:
  Whereas both PAV and Phragmen's sequential rule are considered compelling rules for approval-based multi-winner voting, the latter is more suitable in being extended to more general settings such as PB. Recalling our research motivation from Section~\ref{section:introduction},
where we stated our aim at finding the right axiom and corresponding rule for PB,
as a conclusion we can say that BPJR-L appears to be a compelling axiom for proportional representation of PB and GPseq seems to be a particularly useful and desirable rule in this context. 

We envisage further work on axiomatic and computational aspects of participatory budgeting. It will be interesting to explore the trade-offs between the axioms we proposed and other axioms that might be important to consider in these applications. It may also be useful to, theoretically and empirically, compare how different rules in the literature fare in terms of axiomatic properties. 



\bibliographystyle{plainnat}


%


\end{document}